\documentclass[10pt]{article}

\usepackage{graphicx}
\usepackage{amsmath,amssymb}
\usepackage{amsmath}
\usepackage{amssymb}
\usepackage{amsfonts}
\usepackage{tabularx}
\usepackage{amsthm}
\usepackage{cite}
\usepackage{multirow}
\usepackage{extarrows}

\DeclareMathOperator{\Ex}{Ex}
\DeclareMathOperator{\Var}{Var}

\DeclareMathOperator{\Exp}{Exp}

\newcommand{\N}{{\mathbb N}}

\newcommand{\R}{{\mathbb R}}
\newcommand{\Z}{{\mathbb Z}}

\newcommand{\HH}{{\cal H}}

\newcommand{\DTQW}[1]{X^{(D)}_{#1}}
\newcommand{\CTQW}[1]{X^{(C)}_{#1}}

\newcommand{\FTD}[1]{X^{(F)}_{#1}}
\newcommand{\ftd}{{\tilde t}}

\newcommand{\LA}[1]{Z^{(F)}_{#1}}

\newcommand{\IND}{{\chi}}
\newcommand{\ID}{{\bf 1}}

\newcommand{\df}{\delta }

\newcommand{\ket}[1]{|{#1}\rangle}
\newcommand{\kb}[1]{|{#1}\rangle\!\langle{#1}|}

\newcommand{\bra}[1]{\langle{#1}|}
\newcommand{\bkt}[2]{\langle{#1}|{#2}\rangle}

\newcommand{\wv}[2]{\langle{#1}\rangle_{#2}}

\newtheorem{defin}{Definition}[section]
\newtheorem{thm}{Theorem}[section]
\newtheorem{col}[thm]{Corollary}
\newtheorem{lem}[thm]{Lemma}

\begin{document}
\title{From Discrete Time Quantum Walk to Continuous Time Quantum Walk in Limit Distribution}
\author{Yutaka Shikano$^{1,2}$~\thanks{email: yshikano@ims.ac.jp} 
\\ $^1$Institute for Molecular Science, Okazaki, Aichi, Japan. \\
$^2$Department of Physics, Tokyo Institute of Technology, Tokyo, Japan.}
\date{\today}
\maketitle
\begin{abstract}
The discrete time quantum walk defined as a quantum-mechanical analogue of the discrete time random walk have recently been attracted from various and interdisciplinary fields. In this review, the weak limit theorem, that is, 
the asymptotic behavior, of the one-dimensional discrete time quantum walk is analytically shown. From the limit distribution of the discrete time quantum walk, the discrete time quantum walk can be taken as the quantum dynamical simulator of some physical systems.  
\end{abstract}
\tableofcontents
\section{Introduction} \label{intro}
	A discrete time quantum walk is defined as the quantum-mechanical analogue of the 
	discrete time random walk but is not the quantization of the random walk. Since this description is 
	stroboscopic and unitary, this can be taken as a simple quantum Turing machine. Therefore, there 
	have recently been several theoretical proposals on quantum information processing and 
	many experimental demonstrations in some physical systems. On the other hand, the discrete time 
	quantum walk also attracts mathematicians to better understand the stochastic process. The aim 
	of this review is to connect the mathematical treatments and physical implementations and meanings 
	of the discrete time quantum walk. 
	
	First of all, we give a rough explanation of the discrete time quantum walk. 
	As the random walker, which has information of the position, moves to the left or to the right site depending on the result 
	of a coin flip, a quantum walker and a quantum coin are defined as a position quantum state 
	and a two-level quantum state, which is called a qubit, respectively. It is noted that this 
	formulation is different from the quantization of the discrete time random walk. 
	One step operations of the discrete time random and quantum walks are summarized in Fig.~\ref{qwrwstep}. 
	Furthermore, a quantum coin flip and a shift should be unitary operations because these processes 
	are subject to quantum mechanics. It is emphasized that the coin flip is replaced by the one-qubit 
	operation and the shift operation keeps the quantum coherence between the position and 
	the coin summarized as a quantum circuit representation in Fig.~\ref{dtqwcircuit}. 
	This mathematical definition will be stated in Sec.~\ref{dtqw_sec}.
\begin{figure}[h]
	\begin{center}
	\includegraphics[width=11cm]{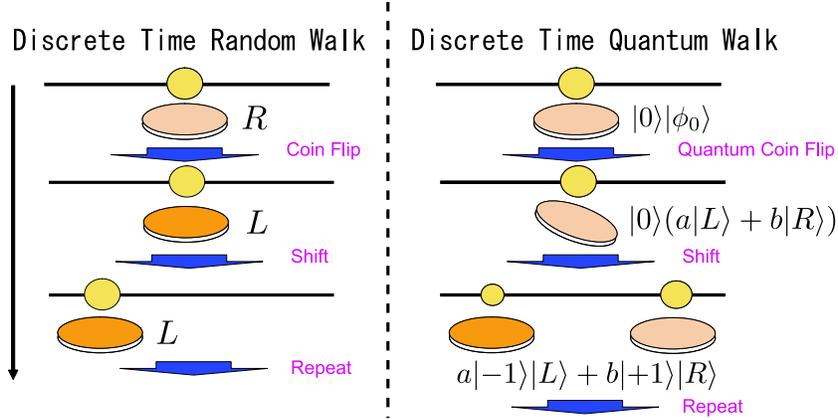}
	\caption{One step of the discrete time random walk v.s. quantum walk. The initial state 
	$\ket{0} \ket{\phi_0}$ expresses that the quantum walker has the quantum coin as $\ket{\phi_0}$ 
	at the origin $\ket{0}$. The coin space is spanned by the two orthogonal states denoted as 
	$\ket{L}$ and $\ket{R}$.}
	\label{qwrwstep}
	\end{center}
\end{figure}
\begin{figure}[h]
	\begin{center}
	\includegraphics[width=8cm]{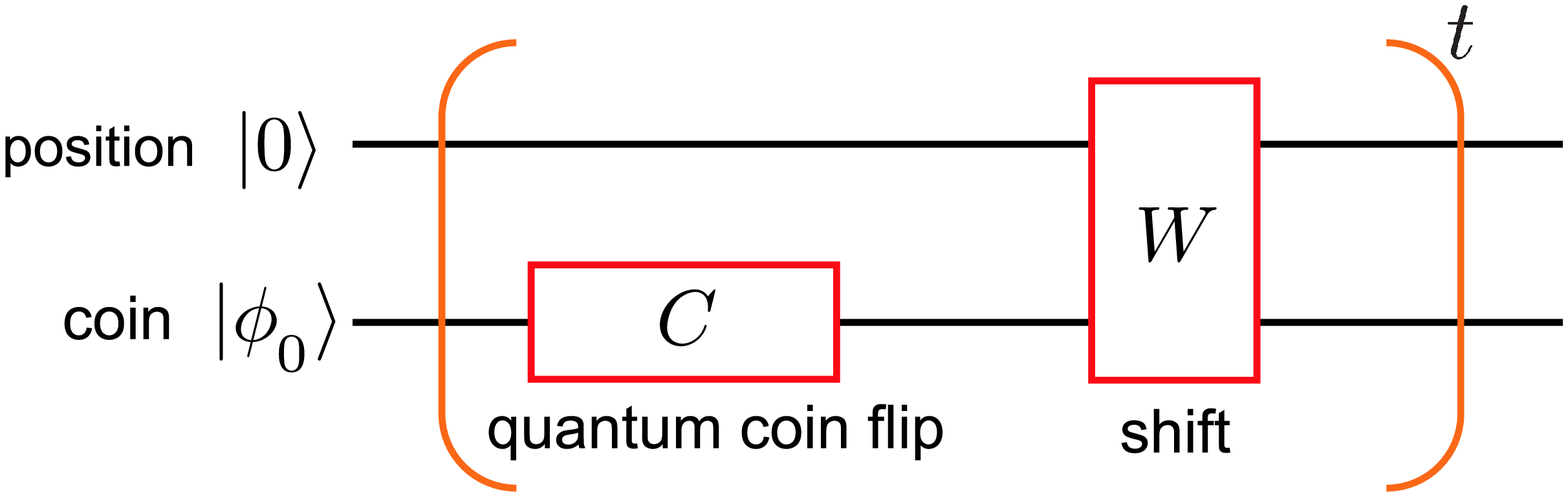}
	\caption{Quantum circuit representation of the discrete time quantum walk. The quantum 
	coin flip and the shift are denoted as $C$ and $W$, respectively. The probability 
	distribution at $t$th step is given by the position measurement after the partial 
	trace over the quantum coin.}
	\label{dtqwcircuit}
	\end{center}
\end{figure}

	As the mathematical motivation, for the asymptotic behaviors, how different between 
	the quantum and classical random walk has been not yet understood in the discrete and 
	continuous cases. While the discrete time quantum walk is experimentally realized 
	under the motivations to realize the primitive of quantum computation, see for more 
	details in Sec.~\ref{dtqw_sec}, we have not yet understood how robust the discrete 
	time quantum walk is under the influence of the noise. In this review, we analytically derive 
	the limit distribution of the several discrete time quantum walks in the one-dimensional system. 
	We show that the discrete time quantum walk can simulate various quantum dynamics.
\section{Review of Discrete Time Quantum Walk} \label{dtqw_sec}
Throughout this review, we focus on
a one-dimensional discrete time quantum walk (DTQW) with two-dimensional coins.
The DTQW is defined as a quantum-mechanical analogue of the classical random walk. 
The Hilbert space of the system is a tensor product 
$\HH_p \otimes \HH_c$, where $\HH_p$ is the position space of a quantum walker
spanned by the complete orthonormal basis $\{ \ket{n} \}$ ($n \in \mathbb{Z}$) and $\HH_c$ is the coin Hilbert space 
spanned by the two orthonormal states $\ket{L} = ( 1 , 0 )^{{\bf T}}$ and $\ket{R} = ( 0 , 1 )^{{\bf T}}$. 
A one-step dynamics is described by a unitary operator $U_t = WC_t$ with 
\begin{align}
	C_t & := \sum_n \left[ (a_{n,t} \ket{n,L} + c_{n,t} \ket{n,R})\bra{n,L} + (d_{n,t} 
	\ket{n,R} + b_{n,t} \ket{n,L}) \bra{n,R} \right], \label{coinop} \\
	& = \sum_n \left[ \kb{n} \otimes \left( \begin{array}{cc} a_{n,t} & b_{n,t} \\ c_{n,t} & d_{n,t} \end{array} \right) \right] \notag \\
	& = \sum_n \left[ \kb{n} \otimes \hat{C}_{n,t} \right] \notag \\
	W & := \sum_n \left( |n-1,L \rangle \langle n,L| + |n+1,R \rangle \langle n,R| \right), 
	\label{shift}
\end{align} 
where $\ket{n, \xi} =: \ket{n} \otimes \ket{\xi} \in \HH_p \otimes \HH_c \ (\xi = L, R)$
and the coefficients at each position satisfy the following relations: 
 $|a_{n,t}|^2 + |c_{n,t}|^2 = 1$, $a_{n,t} \overline{b}_{n,t} + c_{n,t} \overline{d}_{n,t} = 0$, $c_{n,t} = 
- \Delta_{n,t} \overline{b}_{n,t}$, 
$d_{n,t} = \Delta_{n,t} \overline{a}_{n,t}$, where $\Delta_{n,t} = a_{n,t} d_{n,t} - b_{n,t} c_{n,t}$ with $ |\Delta_{n,t}|= 1$ for any $t$.  
Two operators $C_t$ and $W$ are called coin and shift operators, respectively. 
The probability distribution at the position $n$ at the $t$th step is then defined by 
\begin{equation}
	\Pr [n;t] = \sum_{\xi \in \{ L,R \} } \left| \bra{n,\xi} \prod_t U_t \ket{0,\phi_0} \right|^2.
\end{equation}
Mathematically speaking, the position of the DTQW at $t$th step is a random variable denoted as $X_t$.
It is remarked that a time-independent and homogeneous version of this DTQW was first introduced in Ref.~\cite{Ambainis}. 
Throughout this review, we only consider the cases of 
the time-independent and homogeneous coin, the time-dependent and homogeneous coin, the 
time-independent and inhomogeneous coin.
Therefore, the quantum coin is often described as $\hat{C}_{n,t} = C$ or $\hat{C}_{n,t} = \hat{C}_n$ 
in this review. As a simple example, we show the dynamics of the DTQW with the Hadamard coin given by 
\begin{equation} \label{hcoin}
	\hat{C}_{n,t} = \frac{1}{\sqrt{2}} \left( \begin{array}{cc} 1 & 1 \\ 1 & -1 \end{array} \right) 
\end{equation}
by the third step in the case of the initial coin state $\ket{\phi_0} = \ket{L}$ illustrated 
in Fig.~\ref{dqo_fig} and at $1000$th step in the case of the initial state 
$\ket{\phi_0} = (\ket{L} + i \ket{R})/\sqrt{2}$ illustrated in Fig.~\ref{2000thstep}.
\begin{figure}[ht]
	\centering
	\includegraphics[width=10cm]{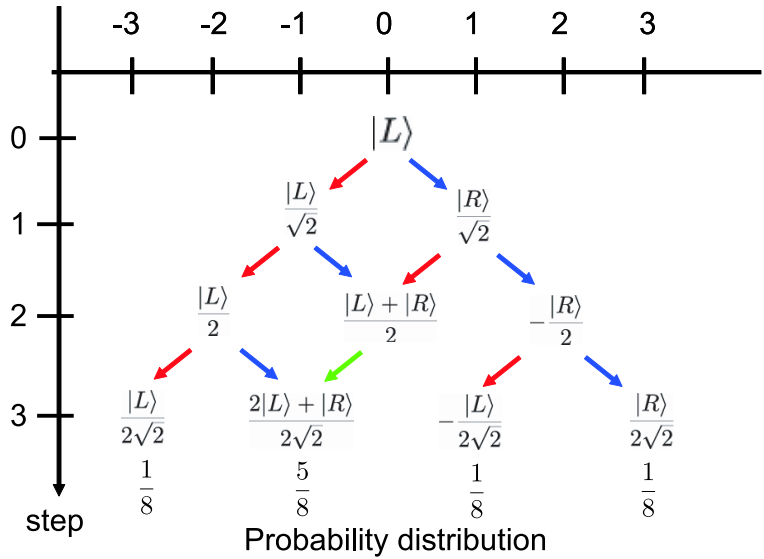}
	\caption{Dynamics of the DTQW with the Hadamard coin (\ref{hcoin}) by the third step. 
	The coin state is expressed at each position. The initial coin state is $\ket{L}$ at the 
	origin. At the third step, the quantum interference occurred at the origin to 
	be indicated at the green arrow. Therefore, 
	the DTQW does not have the spatial symmetry like the classical random walk. This is dependent of 
	the coin operator and the initial coin state.}
	\label{dqo_fig}
\end{figure}
\begin{figure}[ht]
	\centering
	\includegraphics[width=13cm]{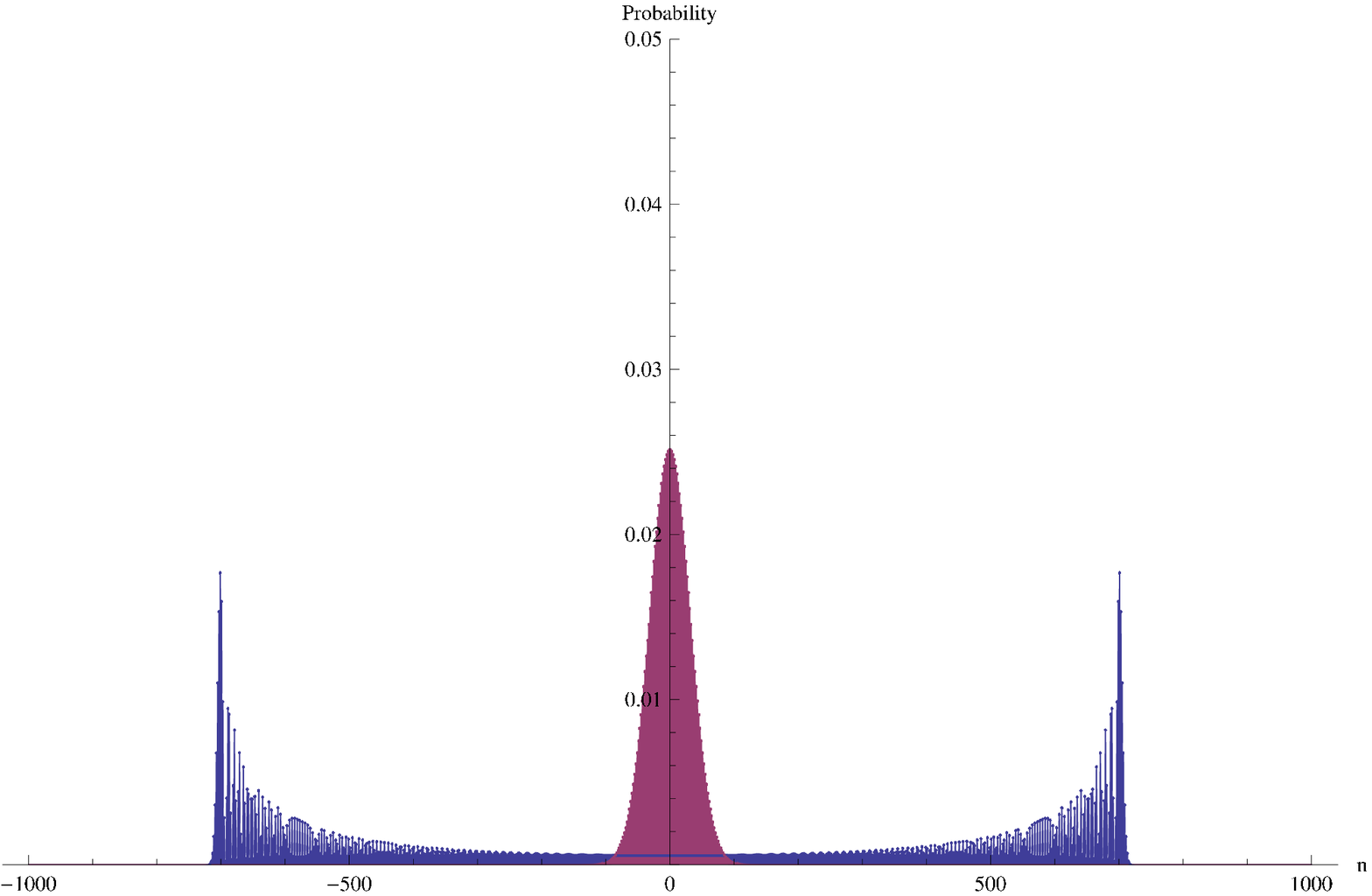}
	\caption{Probability distribution of the DTQW with the Hadamard coin (\ref{hcoin}) 
	at the $1000$th step with the initial coin state $(\ket{L} + i \ket{R})/\sqrt{2}$.}
	\label{2000thstep}
\end{figure}

Historically speaking, the concept of the DTQW was introduced in the independent 
seminal works in mathematics, physics, and 
information science. As far as the author knows, Gudder first introduced its concept 
in his book~\cite{Gudder}. 
While his motivation was to define the stochastic process in quantum probability, 
he did not mathematically clearly 
defined the concept in his book. Note that, curiously, he introduced the three step dynamics of the DTQW 
as the candidate of the stochastic 
process in quantum probability. However, this is an interesting open question. 
Second, Aharonov {\it et al.} 
introduced a slightly different version of the DTQW as {\it quantum random walk}~\cite{AharonovQW}. 
Their quantum walk is defined as follows. 
While one-step dynamics is identical to one of the DTQW, in their introduced quantum walk, 
for every step, the coin state has to be replaced by the initial one. This process is 
a non-unitary process. While they introduced this as 
the quantum-mechanical analogue of the classical random walk, their ultimate aim 
is to construct a process to extract the weak value. 
Since for every step, the same coin is used, we can extract the weak value of 
the spin component, for example, $\wv{\sigma_z}{w}$ as the difference of the shift of the 
position of a quantum particle. This situation mimics the 
Stern-Gerlach gedankenexperiment. Finally, from the viewpoint of information science, 
Meyer tried to construct a quantum-mechanical analogue of the cellular automaton to 
describe dynamics of the one-dimensional quantum gas~\cite{Meyer}. 
He showed no-go lemma that keeping the unitarity, a quantum-mechanical analogue of the random walk cannot 
be described by a single-component complex-valued function; physically speaking, this is the wavefunction 
without the spin component. 

Since the DTQW is defined as the quantum-mechanical analogue of the classical random walk, the DTQW 
has many applications to some systems like the classical random walk as the theoretical studies. 
For examples, in quantum information science, the DTQW is used as the quantum-speedup 
algorithms~\cite{Ambainis2,Buhrman,Magniez,Magniez2}. The Grover algorithm, which is the 
spatially search algorithm, can be taken as the 
tool of the DTQW~\cite{Shenvi}. Furthermore, the Lieb-Robinson bound, which is the upper bound of 
the quantum-speedup algorithm with the energy gap, is related to the DTQW in Ref.~\cite{Chandrashekar10}. 
The DTQW on a graph is a primitive of universal quantum 
computation~\cite{Lovett}~\footnote{In the case of the continuous 
time quantum walk, this is also satisfied~\cite{Childs2}.}. See more details in the book~\cite{VA} and the paper~\cite{VA2}.
Furthermore, the entanglement of the DTQW has been analyzed in 
Refs.~\cite{Chandrashekar06, Maloyer07, Abal07, Liu11} for the coin and space state of a single particle and 
in Refs.~\cite{Rohde11, Stefanak11} for multi particles. Also, 
the DTQW is used as the quantum simulator on the quantum phase transition from 
the Mott insulator to the super-fluid phase, which means the long-range coherence~\cite{Chandrashekar08}, 
the Landau-Zehner transition~\cite{Oka05}, and the classification of 
the topological phase~\cite{Kitagawa10}. Recently, there are the theoretical progresses on the multi-particle 
DTQW for the non-interacting case~\cite{Mayer11, GASM} and the 
interacting one~\cite{Ahlbrecht11, Lahini11}. On the physical properties on the DTQW, see 
the review papers~\cite{Kempe03, Kendon07, Kitagawa}.
As the mathematical aspects, there are many analytical 
studies on the asymptotic behaviors on the DTQW to see more details in the book~\cite{KonnoRev}.
Almost all classes of the DTQW except for a one-dimensional one with a two-dimensional 
coin~\cite{Konno02, Konno05, Grimmett} have the aspect of 
the localization defined that the limit distribution of the DTQW 
divided by some power of the time variable 
has the probability density given by the Dirac delta function. 
In the one dimensional DTQW, the localization was shown in
the models with a three-dimensional coin~\cite{Inui}, a four-dimensional coin~\cite{Inui2}, 
a two-dimensional coin with memory~\cite{McGettrick}, 
a two-dimensional coin in a random environment~\cite{Joye}, two-dimensional coins~\cite{Liu}, 
an spatially incommensurate coin~\cite{SK} and a time-dependent two-dimensional coin 
on the Fibonacci quantum walk~\cite{Romanelli} and 
a random coin~\cite{Joye2, Obuse}. On the other hand, the nature of the localization in 
the two-dimensional DTQW has been studied numerically~\cite{Mackay, Tregenna} and 
analytically~\cite{Inui3, Watabe}. The DTQW on the Cayley tree with a multi-state 
coin was also studied~\cite{Chisaki, Chisaki2}. Furthermore, the recurrence properties of the DTQW, 
i.e., the decay rate of the probability at the origin, were studied 
in Refs. \cite{Wojcik,Stefanak,Stefanak2,Stefanak3,Chandrashekar2,Xu}. 
Recently, the localization property of the DTQW with the inhomogeneous two-dimensional coin has been 
analytically shown in the model in which the coin at the origin is 
different from the rest~\cite{Konno1,Konno2} and in the model of 
the periodically inhomogeneous coins~\cite{Linden}.

It is curiously pointed out that the DTQW has not been found in natural phenomena while the classical 
random walk was used as the model of the geometrical structure of the complex polymer, 
"for more details see the book~\cite{random_book}. However, the DTQW has been 
experimentally realized. In the system of the molecule with the nuclear magnetic resonance, 
there is the experimental realization by third step~\cite{Ryan}. 
In the system of the ion trap, there are theoretical proposals 
and their experimental realization with $\,^{40}$Ca$^{+}$ ion~\cite{Roos} to 
be motivated by the quantum simulator of the Dirac equation~\cite{Roos2}. 
In the system of the optical lattice, there are theoretical proposals 
and their experimental realizations~\cite{Eckert, Karski}. Furthermore, there are experimental 
demonstration of the DTQW by three step~\cite{Schreiber}, of the DTQW with decoherence~\cite{Broome}, 
and on the topological phase~\cite{Kitagawa11} in the optical system~\footnote{Here, we listed up 
the experimental realizations of the DTQW. Therefore, Ref.~\cite{Peruzzo} to experimentally 
demonstrate the two-particle continuous time quantum walk is omitted.}.

\section{Limit Distribution of Discrete Time Quantum Walk}
In the following, we focus on the asymptotic behavior of the DTQW. This is because the asymptotic 
behavior is easily computed and is the most fundamental property of the stochastic process like taking 
the thermodynamic limit in statistical mechanics. Studying the asymptotic behaviors of the stochastic 
process is similar to study the connection between the microscopic and macroscopic descriptions. 
From the viewpoint of probability theory, we obtain the weak limit theorem {\it a la} 
the central limit theorem of the classical random walk (\ref{central}). 
\begin{thm}[Konno~\cite{Konno02,Konno05}] \label{konno_proof}
Let $\DTQW{t}$ be the random variable of the homogeneous and time-independent DTQW at time $t$ with the coin operator 
\begin{equation}
	C = \left( \begin{array}{cc} a & b \\ c & d \end{array} \right)
\end{equation} 
and the initial coin state 
\begin{equation}
	\ket{\phi_0} = q_L \ket{L} + q_R \ket{R}
\end{equation}
with $q_L, q_R \in \mathbb{C}$. 
Therefore, we obtain 
\begin{equation} \label{dtqwlimit}
	\frac{\DTQW{t}}{t} \Rightarrow K (|a|), 
\end{equation}
which $\Rightarrow$ means the weak convergence~\footnote{The moment-generating function of the 
random variable under the specific time scale can be converged.}.
Here, $K(r)$ has the probability density function $f (x ; r)$ with a parameter $r \in (0,1)$: 
\begin{equation}
	f (x ; r) = \frac{\sqrt{1-r^2} \, \IND_{(-r,r)}(x)}{\pi(1-x^2)\sqrt{r^2-x^2}} 
	\left\{ 1 - \left( |q_L|^2 - |q_R|^2 + 
	\frac{a q_L \overline{b q_R} + \overline{a q_L} b q_R}{|a|^2} \right) x \right\},
\end{equation} 
where the indicator function $\IND_{(-r,r)}(x)$ is defined as 
	\begin{equation} \label{indf}
		\IND_A (x) = \begin{cases} 1 \ (x \in A) \\ 0 \ (x \notin A) \end{cases}.
	\end{equation}
\end{thm}
While Konno used the combinatorial method in this proof~\cite{Konno02,Konno05}, 
in the following, we prove this by the spatial Fourier transformation~\cite{Grimmett}.
\begin{proof}
	Taking the Fourier transform for the state $\ket{n, \xi}$ with the position $n \in \Z$ and the 
	coin state $\ket{\xi} \in \HH_c$, we have 
	\begin{equation}
		\ket{k, \hat{\xi}} = \sum_{n \in \Z} e^{-i k n} \ket{n, \xi}.
	\end{equation}
	Since the time evolution in the phase space is given by 
	\begin{align}
		\ket{\hat{\Psi} (k)_t} & := \sum_{n \in \Z} e^{-i k n} U^{t} \ket{0, \phi_0} \notag \\
		& = e^{ik} \begin{bmatrix} a & b \\ 0 & 0 \end{bmatrix} \sum_{n \in \Z} e^{- i k (n + 1)} 
		\bra{n+1} U^{t-1} \ket{0, \phi_0} \notag \\ & \ \ \ + e^{-ik} \begin{bmatrix} 0 & 0 \\ c & d \end{bmatrix} 
		\sum_{n \in \Z} e^{- i k (n - 1)} \bra{n-1} U^{t-1} \ket{0, \phi_0} \notag \\
		& = \begin{bmatrix} e^{ik} & 0 \\ 0 & e^{-ik} \end{bmatrix} U \ket{\hat{\Psi} (k)_{t-1}}
	\end{align}
	and the initial state in the phase space is given by 
	\begin{equation}
		\ket{\hat{\Psi} (k)_0} = q_L \ket{k, \hat{L}} + q_R \ket{k, \hat{R}},
	\end{equation}
	we obtain 
	\begin{equation}
		\ket{\hat{\Psi} (k)_t} = U(k)^t \ket{\hat{\Psi} (k)_0},
	\end{equation}
	where 
	\begin{equation} \label{hen}
		U(k) := \begin{bmatrix} e^{ik} & 0 \\ 0 & e^{-ik} \end{bmatrix} U.
	\end{equation}
	Here, we obtain the probability distribution in the real space, 
	\begin{equation}
		\Pr [n; t] = \int^{\pi}_{- \pi} \frac{d k^{\prime}}{2 \pi} \int^{\pi}_{- \pi} 
		\frac{d k}{2 \pi} e^{i (k - k^{\prime})n} \bra{\hat{\Psi} (k^{\prime})_0} U^{\dagger}
		(k^{\prime})^{t} U(k)^t \ket{\hat{\Psi} (k)_0}.
	\end{equation}
	Therefore, the $j$-th moment for the random variable $X_t$ is given by 
	\begin{align}
		\Ex [(X_t)^j] &:= \sum_{n \in \Z} n^j \Pr [X_t = n] \notag \\ 
		&= \sum_{n \in \Z} \int^{\pi}_{- \pi} \frac{d k^{\prime}}{2 \pi} 
		e^{- i k^{\prime} n} \bra{\hat{\Psi} (k^{\prime})_t} \int^{\pi}_{- \pi} 
		\frac{d k}{2 \pi} \left\{ \left( - i \frac{d}{dk} \right)^j e^{i k^{\prime} n} \right\} 
		\ket{\hat{\Psi} (k)_t} \notag \\
		&= \int^{\pi}_{-\pi} \frac{dk}{2 \pi} \bra{\hat{\Psi} (k)_t} \left( i \frac{d}{dk} \right)^j 
		\ket{\hat{\Psi} (k)_t}. \label{expe}
	\end{align}
	Here, in the last line, we take the partial integration and use the periodic function 
	for $\ket{\hat{\Psi} (k)_t}$.
	
	Let the eigenvalue and its eigenvector of $U(k)$ be denoted as $\lambda_\xi (k)$ 
	and $\ket{v_\xi (k)}$ with $\xi \in \{ 1, 2 \}$, respectively. We have 
	\begin{equation}
		\left( i \frac{d}{dk} \right)^j \ket{\hat{\Psi} (k)_t} = 
		\sum_{\xi \in \{ 1, 2 \}} t (t-1) \cdots (t - j + 1) \lambda_\xi (k)^{t-j} 
		\bkt{v_\xi (k)}{\hat{\Psi} (k)_0} \ket{v_\xi (k)} + O (t^{j-1}).
	\end{equation}
	Therefore, Eq.~(\ref{expe}) is expressed as 
	\begin{align}
		\Ex [(X_t)^j] & = \int^{\pi}_{-\pi} \sum_{\xi \in \{ 1, 2 \}} t (t-1) \cdots (t - j + 1) 
		\lambda_\xi (k)^{-j}
		\left( i \frac{d}{dk} \lambda_\xi (k) \right)^j \times \notag \\ 
		& \ \ \ \ \ \ \ \ \ \ \ \ \ \ \ \ \ \ \ \ \ \ \ \ \ \ \ \ \ \ \ \ \ \ \ \ \ \ \ \ \ \ \ \ \ \ 
		\bkt{v_\xi (k)}{\hat{\Psi} (k)_0} 
		\bkt{\hat{\Psi} (k)_0}{v_\xi (k)} + O (t^{-1}) \notag \\
		& = \int^{\pi}_{-\pi} \sum_{\xi \in \{ 1, 2 \}} t (t-1) \cdots (t - j + 1) \left( 
		\frac{i \frac{d}{dk} \lambda_\xi (k)}{\lambda_\xi (k)} \right)^j | 
		\bkt{v_\xi (k)}{\hat{\Psi} (k)_0}|^2 \frac{dk}{2 \pi} + O (t^{-1}).
	\end{align}
	Therefore, we obtain, as $t \to \infty$,
	\begin{equation} \label{prode}
		\Ex \left[ \left( \frac{X_t}{t} \right)^j \right] \to \int^{\pi}_{-\pi} 
		\sum_{\xi \in \{ 1, 2 \}}\left( 
		\frac{i \frac{d}{dk} \lambda_\xi (k)}{\lambda_\xi (k)} \right)^j | 
		\bkt{v_\xi (k)}{\hat{\Psi} (k)_0}|^2 \frac{dk}{2 \pi} 
	\end{equation}
	to obtain the momentum-generating function as 
	\begin{equation}
		\Ex \left[ e^{( i \tau X_t) / t}  \right] \to \sum_{\xi \in \{ 1, 2 \}} \int^{\pi}_{- \pi} \exp \left( i 
		\tau \frac{i \frac{d}{dk} \lambda_\xi (k)}{\lambda_\xi (k)} \right) 
		| \bkt{v_\xi (k)}{\hat{\Psi} (k)_0}|^2 \frac{dk}{2 \pi} \ {\rm as} \ t \to \infty \label{dtqwmomentum}
	\end{equation}
	for any real parameter $\tau \in \R$ since $\exp (i \tau f(k) )$ is regular when $f(k)$ is a continuous function.
	We calculate the eigenvalue of Eq.~(\ref{hen}) to obtain the term of 
	$\frac{i \frac{d}{dk} \lambda_\xi (k)}{\lambda_\xi (k)}$ as 
	\begin{equation}
		\frac{i \frac{d}{dk} \lambda_\xi (k)}{\lambda_\xi (k)} = \pm \frac{|a| \cos \left( k - \frac{\theta}{2} \right)}{
		\sqrt{|a|^2 \cos^2 \left( k - \frac{\theta}{2} \right) - (1 - |a|^2) e^{i 2 k}}} =: \pm x,
	\end{equation}
	where $e^{i \theta} := ad - bc$ due to the unitary condition $|ad - bc| = 1$.
	From Eq. (\ref{dtqwmomentum}) and $| \bkt{v_1 (k)}{\hat{\Psi} (k)_0}|^2 + | \bkt{v_2 (k)}{\hat{\Psi} (k)_0}|^2 = 1$, 
	we obtain 
	\begin{equation}
	\Pr [X_t \leq x] = \int^{k(x)}_{-k(x)} | \bkt{v_1 (k)}{\hat{\Psi} (k)_0}|^2 \frac{dk}{2 \pi} 
	+ \left( \int^{-\pi + k(x)}_{- \pi} + \int^{\pi}_{\pi - k(x)} \right) | \bkt{v_2 (k)}{\hat{\Psi} (k)_0}|^2 \frac{dk}{2 \pi}
	\end{equation}
	Therefore, we obtain the probability density as 
	\begin{align}
		f(x) & = \frac{d}{dx} \Pr [X_t \leq x] \notag \\
		& = \frac{1}{2 \pi} \left( | \bkt{v_1 (k)}{\hat{\Psi} (k)_0}|^2 + 
		| \bkt{v_1 (-k)}{\hat{\Psi} (-k)_0}|^2 \right. \notag \\ & \left. \ \ \ \ \ \ \ \ \ \ 
		+ | \bkt{v_2 (- \pi + k)}{\hat{\Psi} (- \pi + k)_0}|^2 + 
		| \bkt{v_2 (\pi - k)}{\hat{\Psi} (\pi - k)_0}|^2 \right) \frac{dk(x)}{dx} \notag \\
		& = \frac{\sqrt{1-|a|^2} \, \IND_{(-|a|,|a|)}(x)}{\pi(1-x^2)\sqrt{|a|^2-x^2}} 
		\left\{ 1 - \left( |q_L|^2 - |q_R|^2 + 
		\frac{a q_L \overline{b q_R} + \overline{a q_L} b q_R}{|a|^2} \right) x \right\}.
	\end{align}
	This is the desired result.
\end{proof}
The above theorem tells us that the DTQW has the capacities of the quantum speedup since the scaling of the 
random variable of the DTQW is $t$ while one of the classical case is $\sqrt{t}$. This effect originates from the 
superposition of the quantum walker. Also, the effect of the quantum interference causes a inverted bell shape 
distribution, which is completely different from the normal distribution.
\section{Connection to Dirac Equation} \label{demap_sec}
In the following, we derive the free Dirac equation from the DTQW with the specified coin operator;
\begin{equation}
	C (\epsilon) = \left( \begin{array}{cc} \cos \epsilon & - i \sin \epsilon \\ - i \sin \epsilon & \cos \epsilon \end{array} \right)
\end{equation}
with $\epsilon \in \mathbb{R}$. Let $\Psi_n (x) \in \HH_p \otimes \HH_c$ be a quantum state at a position $x \in \mathbb{R}$ with the 
associated coin state $\ket{\psi_x}$ at $n$ step. When we take the tiny parameter $\epsilon$ as one lattice size, 
we obtain the dynamics of the DTQW as 
\begin{align}
\Psi_n (x) & = Q (\epsilon) \Psi_{n-1} (x - \epsilon) + P (\epsilon) \Psi_{n-1} (x + \epsilon) \notag \\ 
& = Q(\epsilon) \left(1 - \epsilon \frac{\partial}{\partial x} + O(\epsilon^2) \right) \Psi_{n-1} (x) 
+ P(\epsilon) \left(1 + \epsilon \frac{\partial}{\partial x} + O(\epsilon^2) \right) \Psi_{n-1} (x)
\end{align}
where
\begin{align}
Q (\epsilon) & = \left( \begin{array}{cc} 0 & 0 \\ 0 & 1 \end{array} \right) - i \epsilon 
\left( \begin{array}{cc} 0 & 0 \\ 1 & 0 \end{array} \right) + O(\epsilon^2) \\ 
P (\epsilon) & = \left( \begin{array}{cc} 1 & 0 \\ 0 & 0 \end{array} \right) - i \epsilon \left( \begin{array}{cc} 0 & 1 \\ 0 & 0 
\end{array} \right) + O(\epsilon^2).
\end{align} 
We obtain the dynamics of the DTQW from the beginning~\cite{Strauch, Bracken}, 
\begin{equation}
\Psi_\tau (x) \sim e^{- i \left( \sigma_x + \sigma_z \frac{\partial}{\partial x} \right) t} \Psi_{0} (x), 
\end{equation}
where $t = \epsilon \tau$. Partially differentiating this with respect to $t$, we obtain
\begin{align}
	\frac{\partial \Psi_t (x)}{\partial t} & \sim - i \left( \sigma_x + \sigma_z \frac{\partial}{\partial x} \right) 
	e^{- i \left(\sigma_x + \sigma_z \frac{\partial}{\partial x} \right) t} \Psi_{0} (x) \notag \\
	& = - i \left( \sigma_x + \sigma_z \frac{\partial}{\partial x} \right) \Psi_t (x).
\end{align}
This equation corresponds to the Dirac equation in $1+1$ dimensions by taking the coin state as the 
spinor~\footnote{This analysis can be extended to the two-dimensional space~\cite{Sato}.}. It is emphasized that 
the DTQW is described in non-relativistic quantum mechanics and the Dirac equation is the fundamental equation 
of relativistic quantum mechanics. This extension originates from the discretization of non-relativistic quantum mechanics. 
As a naive extension, some kinds of differential equations may be simulated from the DTQW model. 
\section{Connection to Continuous Time Quantum Walk} \label{ctqw_sec}
In this section, we will see the relationship between the continuous time quantum walk (CTQW), 
which will be defined in the following subsection, and the DTQW according to Ref.~\cite{CKSS2}.
\subsection{Review of continuous time quantum walk}
Let us define the CTQW~\cite{Farhi}, which was originally motivated in the context of the speedup algorithm 
of quantum computation using the graph, as follows.
The state space is defined as the position state space $\HH_p$ only in contrast with the DTQW. 
Let $\Psi_t^{(C)}(x)$ be a quantum state at a position $x \in \mathbb{R}$. 
The time evolution is given by the discretized Schr\"{o}dinger equation: 
\begin{align}\label{defConti}
	-i\frac{\partial \Psi^{(C)}_t(x)}{\partial t} & = 
	\frac{1}{2} \left(\gamma \Psi^{(C)}_t(x-1) + \overline{\gamma} \Psi^{(C)}_t(x+1) \right) 
	\ \ {\rm where} \ t>0, \notag \\
	\Psi_0^{(C)}(x) &= \df(x), 
\end{align}
where $\gamma$ is a complex number. 
Let $\CTQW{t}$ be the random variable of the CTQW at time $t$. The distribution of $\CTQW{t}$ is given 
by $\Pr (\CTQW{t}=x) = | \Psi^{(C)}_{t}(x) |^2$. 
This can be taken as dynamics of the discretized Schr\"{o}dinger equation. Furthermore, from the viewpoint of 
statistical physics, this can be taken as the hopping dynamics of the Hubbard model noting that the CTQW 
describes the single-particle behavior. 

In analogy to the DTQW (\ref{dtqwlimit}), we obtain the weak limit theorem as 
\begin{thm}[Konno~\cite{KonnoQLC}] \label{Clim}
Let $\CTQW{t}$ be the CTQW at time $t$. 
\begin{equation}
	\frac{\CTQW{t}}{t} \Rightarrow Z(\gamma).
\end{equation}
Here, $Z(a)$ has the probability density;
\begin{equation}
	f(x; a) = \frac{\IND_{(-|a|,|a|)}(x)}{\pi \sqrt{|a|^2-x^2}}. \label{arcsine}
\end{equation} 
\end{thm}
This limit distribution (\ref{arcsine}) is called the arcsine distribution since 
the cumulative distribution function is given by 
\begin{equation}
	\Pr [Z(\gamma) \leq x] := \int^{x}_{- \infty} f(y; \gamma) dy = \left\{ \begin{array}{ll} 0 & \ 
	{\rm on} \ x \leq -|\gamma| \\ 
	\frac{1}{\pi} \arcsin \frac{x}{|\gamma|} + \frac{1}{2} & \ {\rm on} \ -|\gamma|< x < |\gamma| \\
	1 & \ {\rm on} \ |\gamma| \leq 1 \end{array}. \right. 
\end{equation}
It is curiously remarked that 
the arcsine distribution (\ref{arcsine}) can be derived as the limit distribution of the 
quantum probability theory with the monotone independence~\cite{Hasebe} while the CTQW is 
independent of the stochastic process in quantum probability theory.
\begin{proof}
	Taking the spatial Fourier  transform 
	\begin{equation}
		\hat{\Psi}_t (k) := \sum_{x \in \Z} \Psi_t (x) e^{ikx},
	\end{equation}
	we obtain the solution of Eq.~(\ref{defConti}); 
	\begin{equation}
		- i \frac{\hat{\Psi}_t (k)}{\partial t} = \frac{1}{2} \left( \gamma e^{ik} \hat{\Psi}_t (k) + \overline{\gamma} e^{-ik} 
		\hat{\Psi}_t (k)\right);
	\end{equation}
	as 
	\begin{equation}
		\hat{\Psi}_t (k) = \Exp \left[ i |\gamma| t \cos (k + \arg (\gamma)) \right].
	\end{equation}
	From the definition of the spatial Fourier transform, the moment-generating function can be expressed as 
	\begin{align}
		\Ex \left[ e^{i \tau \CTQW{t} / t} \right] & = \int^{2 \pi}_{0} \hat{\Psi}^{\dagger}_t (k) \cdot 
		\hat{\Psi}_t \left(k + \frac{\tau}{t} \right) \frac{dk}{2 \pi} \notag \\
		& = \int^{2 \pi}_{0} \Exp \left[ i |\gamma| t \left\{ \cos \left(k + \arg (\gamma) 
		+ \frac{\tau}{t} \right) - 
		\cos (k + \arg (\gamma)) \right\} \right] \frac{dk}{2 \pi} \notag \\
		& = \int^{2 \pi}_{0} \Exp \left[ i |\gamma| t \sin (k + \arg (\gamma)) \frac{\tau}{t} + O \left( \frac{\tau^2}{t^2} \right) 
		\right] \frac{dk}{2 \pi} \notag \\
		& \to \int^{2 \pi}_{0} \Exp \left[ i |\gamma| \tau \sin (k + \arg (\gamma)) \right] \frac{dk}{2 \pi} 
		\ \ {\rm as} \ \ t \to \infty \notag \\
		& = \int^{|\gamma|}_{-|\gamma|} \frac{ \Exp \left[ i \tau x \right] }{ \pi \sqrt{|\gamma|^2 - x^2} } dx \notag \\
		& = \int^{\infty}_{-\infty} e^{i \tau x} \frac{\IND_{(-|\gamma|,|\gamma|)}(x)}{\pi \sqrt{|\gamma|^2-x^2}} dx.
	\end{align}
	Here, in the last second line, we take the transformed variable $x := |\gamma| \sin (k + \arg(\gamma))$. On the range of the 
	integration, since the parameter $k$ is ranged from $0$ to $2 \pi$, the maximum and minimum values of the parameter $x$ are 
	$|\gamma|$ and $-|\gamma|$, respectively. Since the parameter $x$ is doubly counted, 
	we have the last second line multiplied by two.
	Therefore, we obtain the desired result.
\end{proof}
\subsection{Final-time dependent random walk and lazy random walk}
To pursue the aim of this section, we introduce a final-time-dependent (FTD) walks, 
which are modified walks initiated by Strauch \cite{Strauch1}. 
Let $\ftd$ be the final time, that is, a particle keeps walking until the final time comes. 
At first, we show a construction of the CTRW from FTD-DTRWs. 
Secondly, we give CTQWs in some limit of the FTD-DTQW 
which is a quantum analogue of the final time dependent RW. 
Here, we assume the parameter $r(\ftd)>0$ with $r(\ftd) \to 0$ as $\ftd \to \infty$ for the 
FTD-DTRW and FTD-DTQW.

Let us define the FTD-DTRW on $\Z$, 
where the ``final time'' means the time that a particle stops the walk.  
The average of the waiting time of particle movement is $1/r(\ftd)$. 
Here, $r(\ftd)$ is the probability that the particle moves by the final time $\ftd$. 
The number of particle movements in the walk by the final time $\ftd$ is evaluated as $r(\ftd) \ftd$. 
A particle spends her most time being lazy since the rate of movements $r(\ftd)$ tends to $0$ 
as $\ftd \to \infty$. 
This is called a lazy RW. In the following, we will show that the lazy RW with the number of particle's 
movements $\ftd r(\ftd)$ can be taken as the CTRW with the final time $t$ for sufficiently large $\ftd$.

Let $\LA{m}$ be the lazy RW at time $m \in \{0,1,2,\dots, \ftd \}$ defined by 
\begin{equation}\label{lazyeq}
	p_{m}(x)=(1- r(\ftd)) p_{m-1} (x) + \frac{r(\ftd)}{2} \bigg\{p_{m-1}(x-1)+p_{m-1}(x+1)\bigg\}, 
	\ p_0(x)= \delta_{x,0}, 
\end{equation}
where $p_m(x) := \Pr (\LA{m}=x) \ ( x \in \N)$. 
This means that a particle stays at the same place with the probability $1-r(\ftd)$, and 
the particle jumps left or right with probability $1/2$ when the moving opportunity comes 
with probability $r(\ftd)$. 
Put a vector of the probability density as 
\begin{equation}
	\ket{\mu_m} : = \left[ \begin{array}{c} \vdots \\ p_m(-1) \\ p_m(0) \\ p_m(1) \\ \vdots \end{array} 
	\right] \in \HH_p.
\end{equation}
It is noted that this vector $\ket{\mu_m}$ can be expressed as the real-valued vector.
Equation (\ref{lazyeq}) is equivalent to 
\begin{equation}
	\ket{\mu_m} = \left\{ \ID + r(\ftd) \left( \frac{A}{2} - \ID \right) \right\} \ket{\mu_{m-1}}, 
	\ \ket{\mu_0}=\ket{0}, 
\end{equation}
where $A \ket{x} := \ket{x+1} + \ket{x-1}$ for any $x \in \N$. Here, $\ket{\mu_0}$ means 
$p_0 (x) = \delta_{x,0}$. 
Therefore, we have under the assumption of $\ftd r^2(\ftd) \to 0$ as $\ftd \to \infty$, 
\begin{equation} \label{ctct}
	\ket{\mu_\ftd} \sim \Exp \left[ \ftd r(\ftd) \left( \frac{A}{2} - \ID \right) \right] \ket{\mu_0}.
\end{equation}
Replace the particle movements $\ftd r(\ftd)$ with a continuous parameter $t$. 
Partially differentiating Eq.~(\ref{ctct}) with respect to the parameter $t$, 
we obtain 
\begin{align} 
	\frac{\partial}{\partial t} \ket{\mu_t} & \sim \left( \frac{A}{2} - \ID \right) 
	\Exp \left[ t \left( \frac{A}{2} - \ID \right) \right] \ket{\mu_0} \notag \\
	& \sim \left( \frac{A}{2} - \ID \right) \ket{\mu_t}.
\end{align}
Therefore, we can express $p_\ftd (x) \sim m_s (x)$ for sufficiently large $\ftd$, where 
$m_s(x)$ satisfies 
\begin{equation}\label{ContiRW}
	\frac{\partial}{\partial s}m_{s}(x)
	=\frac{1}{2}\left\{m_s (x+1) + m_s (x-1) \right\}-m_s(x), 
	\ m_0(x)= \df (x) \ {\rm on} \ s \leq t.
\end{equation}
This differential equation corresponds to the CTRW with time $t$. 
We have the central limit theorem by 
using the Fourier transform in the following: 
if $\ftd r(\ftd) \to \infty$, as $\ftd \to \infty$, then 
\begin{equation} \label{LazyCentral}
	\frac{\LA{\ftd}}{\sqrt{\ftd r(\ftd)}} \Rightarrow N(0,1) \ {\rm as} \ \ftd \to \infty. 
\end{equation}
It is remarked that this equation can be shown without $\ftd r(\ftd) \to 0$ as $\ftd \to \infty$
In particular, in the case of $r(\ftd) \sim r/\ftd^\alpha$ with 
$0 < r <1$ and $0 \leq \alpha$, we obtain a crossover from the diffusive to the localized spreading: 
as $\ftd \to \infty$, if $0 \leq \alpha <1$, then 
\begin{equation}\label{DRtoCR}
	\frac{\LA{\ftd}}{\sqrt{\ftd^{1-\alpha}}} \Rightarrow N(0,r)
\end{equation}
and if $\alpha=1$, then 
\begin{equation} \label{DRtoCR2}
	\Pr (\LA{\ftd}=n) \sim e^{-r} I_n (r), 
\end{equation}
where $I_\nu(z)$ is the modified Bessel function of order $\nu$; 
\begin{equation}
	I_\nu (z) = \sum^{\infty}_{m = 0} \frac{1}{2^{2 m + \nu} \, \Gamma (m + 1) 
	\Gamma (\nu + m + 1)} z^{2 m + \nu}.
\end{equation} 
Here, the gamma function $\Gamma (z)$ is given by 
\begin{equation}
	\Gamma (z) = \int^{\infty}_{0} y^{1-z} e^{-t} dy. 
\end{equation}
Note that Eq.(\ref{DRtoCR2}) comes from the correspondence between $\ftd r(\ftd)$ and $t$, and 
the Fourier transform for Eq. (\ref{ContiRW}): 
for large $\ftd$, 
\begin{equation}
	p_\ftd (x) \sim m_t(x) = e^{-t}I_x(t), 
\end{equation}
where and if $\alpha>1$, then we can easily see that $\Pr (\LA{\ftd}=x) \to \df (x)$. 
\subsection{Final-time dependent discrete time quantum walk}
In the following, we will consider a quantum-mechanical analogue of the above method of 
continuum approximation to the lazy RW. The FTD coin operator is defined by 
\begin{equation}\label{Qcoin}
	C_\ftd := \begin{bmatrix} \sqrt{r(\ftd)} & \sqrt{1-r(\ftd)}  \\  \sqrt{1-r(\ftd)} & 
	-\sqrt{r(\ftd)}  \end{bmatrix}. 
\end{equation}
This is defined as a quantum-mechanical analogue to the FTD stochastic coin of the correlated 
RW defined~\cite{KonnoCor}. 
According to Ref.~\cite{Strauch1}, the absolute values of diagonal parts of the quantum coin are sufficiently small 
and independent of the final time $\ftd$. According to Ref.~\cite{Romanelli1}, 
the coin operator is changed at each time, and its diagonal parts at time $m(<\ftd)$ 
are given in proportion to $1/m^\alpha$ ($0 \leq \alpha \leq 1$). 
On the other hand, in our model, all elements of the quantum coin depend on the final time $\ftd$. 
The above quantum coin (\ref{Qcoin}) shows that a particle moves in the same direction of the 
previous step with the probability amplitude $\sqrt{r(\ftd)}$ (left case) 
and $-\sqrt{r(\ftd)}$ (right case), and the opposite one with $\sqrt{1-r(\ftd)}$ (left and right cases).
This is called an FTD-DTQW. We give the following lemma which shows that 
the FTD-DTQW is expressed as a linear combination of some CTQWs for sufficiently large $\ftd$. 
The following lemma is consistent to Refs. \cite{Strauch1,Romanelli1} except for the time scaling. 
Let $\ftd$ be the final time for the FTD-DTQW and $t$ be the time for the corresponding CTQWs. 
We define a quantum analogue of the waiting time of a quantum particle movement by $\ftd / t$. 
\begin{lem}\label{DtoC}
	Let $\Psi_n^{(F)}(x)$ be the coin state of the FTD-DTQW at time $\ftd$ and position $x$. 
	Put $t=\ftd \sqrt{r(\ftd)}$ with $\ftd r(\ftd)= o(1)$ for large $\ftd$. 
	$\Psi_\ftd^{(F)}(x)$ is asymptotically expressed by 
	\begin{equation}
		\Psi_\ftd^{(F)}(x) \sim \frac{1}{2}\left( \Psi_t^{(+)}(x)+(-1)^\ftd \Psi_t^{(-)}(x) \right)
	\end{equation}
	with $t=n\sqrt{r(\ftd)}$, where $\Psi_s^{(\pm)}(x)$ satisfies the following Schr\"{o}dinger equation:
	\begin{align}
	-i\frac{\partial}{\partial s} \Psi_s^{(\pm,\xi)}(x) &= \pm \frac{1}{2} \left(i\Psi_s^{(\pm,\xi)}(x-1)
	-i\Psi_s^{(\pm,\xi)}(x+1)\right),\;\; (\xi \in \{L,R\}) \label{qwsch} \\
	\Psi_0^{(\pm,L)}(x) &= q_L \delta_{x,0} \pm q_R \delta_{x,1}, \ \ \Psi_0^{(\pm,R)}(x) = q_R 
	\delta_{x,0} \pm q_L \delta_{x,-1}, \label{qwsch1}
	\end{align} 
	where $\Psi_s^{(\pm,\xi)}(x) =\bkt{\xi}{\Psi_s^{(\pm)}(x)}, (\xi \in \{ R,L \})$. 
\end{lem}
It is noted that Eq. (\ref{qwsch}) can be taken as the CTQW with $\gamma = \pm i$ and the modified 
initial condition Eq. (\ref{qwsch1}).
\begin{proof}
By the spatial Fourier transform for Eq. (\ref{Qcoin}), we obtain 
\begin{equation}
	\widehat{C}_\ftd^2 (k) = \ID - 2i \sqrt{r(\ftd)} \sin k V_{\sigma_x}(k) +O(r(n)), 
\end{equation}
with $V_{\sigma_x}(k) := \left( e^{-i k} \kb{L} +e^{i k} \kb{R} \right){\sigma_x}$. 
When $|\det (\sqrt{r(\ftd)}\sin k V_{\sigma_x}(k))| <1$, so 
\begin{equation}
	\log(\widehat{C}_\ftd^2 (k)) = -2i \sqrt{r(\ftd)}\sin k V_{\sigma_x}(k)+O(r(\ftd)),
\end{equation} 
where $\det (A)$ is the determinant of $A$. 
Therefore, we have under the assumption of $\ftd r(\ftd) \to 0$ as $\ftd \to \infty$, 
\begin{equation} \label{DtoCMaster}
	\widehat{C}^{\ftd}(k)=(V_{\sigma_x}(k))^{n}e^{-i \ftd \sqrt{r(\ftd)}\sin k V_{\sigma_x}(k)+O(\ftd r(\ftd))}
	\sim (V_{\sigma_x}(k))^{\ftd} e^{-i \ftd \sqrt{r(\ftd)}\sin k V_{\sigma_x}(k)}. 
\end{equation}
Because of $(\ID \pm V_{\sigma_x}(k))e^{-i s \sin k V_{\sigma_x}(k)} = 
e^{\mp is\sin k}(\ID \pm V_{\sigma_x}(k))$ 
for any real number $s$, 
we obtain 
\begin{equation}\label{applox1}
	\widehat{C}^{\ftd}(k) \sim \frac{1}{2} \left( e^{-i \ftd \sqrt{r(\ftd)} \sin k } 
	(\ID+V_{\sigma_x}(k))+(-1)^{\ftd} 
	e^{i \ftd \sqrt{r(\ftd)}\sin k} (\ID - V_{\sigma_x}(k)) \right). 
\end{equation}
To see a relationship between the FTD-DTQW and the CTQW, 
we define $\widehat{\Psi}_s^{(\pm)}(k)\equiv e^{\mp is\sin k}\widehat{\Psi}^{(\pm)}_0(k)$ with 
$\widehat{\Psi}^{(\pm)}_0(k)\equiv (\ID \pm V_{\sigma_x}(k)) \ket{\phi_0}$, where 
$\ket{\phi_0}$ is the initial coin state.
It is noted that 
\begin{equation}
	\widehat{\Psi}_{\ftd}(k) = \frac{\widehat{\Psi}_{t}^{(+)}(k)+(-1)^{\ftd} \widehat{\Psi}_{t}^{(-)}(k)}{2},
\end{equation} 
with $t=\ftd \sqrt{r(\ftd)}$ and $\widehat{\Psi}_{s}^{(\pm)}(k)$ obeys 
\begin{equation}
	\pm i \frac{d}{ds}\widehat{\Psi}_s^{(\pm)}(k)=\sin k \widehat{\Psi}_s^{(\pm)}(k).
\end{equation}
\end{proof}
By the inverse Fourier transform for Eq. (\ref{applox1}) and the definition of the Bessel function, 
we get the following theorem. 
\begin{thm}\label{kamakura}
Let $\FTD{\ftd}$ be the FTD-DTQW at the final time $\ftd$. 
Put $t=\ftd \sqrt{r(\ftd)}$. 
When the initial coin state is $\ket{\phi_0} =q_L \ket{L} +q_R \ket{R}$ with $|q_L|^2+|q_R|^2=1$, then we have 
\begin{equation}
	\Pr (\FTD{\ftd}=x) \sim \frac{1+(-1)^{\ftd + x}}{2}\mathcal{J}(x;t), 
\end{equation}
where 
\begin{equation}
	\mathcal{J}(x;t) := \left\{1- (\overline{q_R} q_L + q_R \overline{q_L}) 
	\frac{2x}{t}\right\}J_x^2(t)+|q_L|^2 J_{x-1}^2(t)+|q_R|^2 J_{x+1}^2(t).
\end{equation}
Here $J_\nu(z)$ is the Bessel function of the first kind of order $\nu$ as 
\begin{align}
	J_\nu(z) & = \sum^{\infty}_{m=0} \frac{(-1)^m}{2^{2 m + \nu} \Gamma (m) \Gamma (m + \nu)} z^{2 m + \nu}.
\end{align} 
\end{thm}
The following theorem shows that 
the weak convergence theorem of the FTD-DTQW also holds 
without the assumption of $\ftd r(\ftd) \to 0$ as $\ftd \to \infty$. 
\begin{thm}\label{COQCD}
	Let the initial coin state be $\ket{\phi_0}= q_L \ket{L} +q_R \ket{R}$ with $|q_L|^2+|q_R|^2=1$. 
	Assume $\ftd \sqrt{r(\ftd)} \to \infty$ as $\ftd \to \infty$. 
	Then we have as $\ftd \to \infty$, 
	\begin{equation}\label{WCDtoC}
		\frac{\FTD{\ftd}}{\ftd \sqrt{r(\ftd)}} \Rightarrow A^{\phi_0}(1), 
	\end{equation}
	where $A^{\phi_0}(y)$ has the following density:
	\begin{equation}
		f(x;y) = \frac{ \IND_{(-|y|,|y|)}(x)}{\pi \sqrt{y^2-x^2}} \left\{ 
		1- \frac{(\overline{q_R} q_L + q_R \overline{q_L})x}{y} \right\}. 
	\end{equation} 
\end{thm}
\begin{proof}
	We show that Eq. (\ref{WCDtoC}) holds without assumption $\ftd r(\ftd) \to 0$ as $\ftd \to \infty$. 
	The Fourier transform for the quantum coin $C_\ftd$ (\ref{Qcoin}) is described as 
	\begin{equation}
		\widehat{C}_\ftd (k)=
	\begin{bmatrix} e^{-ik}\sqrt{r(\ftd)} & e^{-ik}\sqrt{1-r(\ftd)} \\ e^{ik}\sqrt{1-r(\ftd)} 
	& -e^{ik}\sqrt{r(\ftd)}  \end{bmatrix}. 
	\end{equation}
Let the eigenvalue and corresponding eigenvector of $\widehat{C}_\ftd (k)$ 
be $e^{i \theta^{(\pm)}_\ftd (k)}$ and $\ket{v^{(\pm)}_\ftd (k)}$. 
Then we obtain 
\begin{align}
\cos \theta^{(\pm)}_n(k) &= \pm \sqrt{1-r(\ftd)\sin ^2 k}, \\ \sin \theta^{(\pm)}_\ftd (k) &=-\sqrt{r(\ftd)}\sin k, \\
\pi^{(\pm)}_\ftd (k) \equiv \kb{v^{(\pm)}_\ftd (k)}  
&= \frac{1}{2}\left( \ID \pm V_{\sigma_x} (k) \right) + O \left( \sqrt{r(\ftd)} \right). 
\end{align}
Let $\FTD{s}$ be the FTD-DTQW at time $s$ with the initial coin state 
$\ket{\phi_0} = q_L \ket{L} + q_R \ket{R}$. It is noted that 
\begin{equation}
	\theta_\ftd^{(\pm)} (k+\xi/t) - \theta_\ftd^{(\pm)}(k)= \mp \xi \frac{\sqrt{r(\ftd)}}{t}\cos k
        +O \left(\frac{(\sqrt{r(\ftd)})^3}{t}\right),
\end{equation}
where $t=\ftd \sqrt{r(\ftd)}$. We have 
\begin{align}
\Ex [e^{i\xi \FTD{\ftd}/t}] 
	&=\int_{0}^{2\pi} \bra{\phi_0} (\widehat{C}^\ftd (k))^\dagger \widehat{C}^\ftd 
	(k+\xi/t) \ket{\phi_0} \frac{dk}{2\pi} \notag \\
	&=\int_{0}^{2\pi} 
        \sum_{j\in \{0,1\}} e^{-(-1)^j i\xi\cos k+O\left(\sqrt{r(\ftd)}\right)} \notag \\
        &\;\;\;\;\;\;\;\;\; \times \frac{1}{2} \bra{\phi_0} \left(\ID +(-1)^j V_{\sigma_x}(k)\right) 
		\ket{\phi_0} \frac{dk}{2\pi} 
		+O(\sqrt{r(\ftd)}) \notag \\
        &\to \int_{-\infty}^{\infty}e^{i \xi x} \bigg\{1-(q_L \overline{q_R}+\overline{q_L}q_R) x \bigg\}
		\frac{\IND_{(-1,1)}(x)}{\pi\sqrt{1-x^2}}dx
        \ {\rm as} \ \ftd \to \infty,
\end{align}
where $t=\ftd \sqrt{r(\ftd)}$. Then, we get the desired conclusion. 
\end{proof}
From Theorems \ref{kamakura} and \ref{COQCD}, we obtain the following limit theorem 
which shows a crossover from the localized to the ballistic spreading. 
\begin{col}\label{DQtoCQ}
If $\sqrt{r(\ftd)}=r/\ftd^{\alpha}$ with $0<r<1$, then as $\ftd \to \infty$, 
\begin{equation}
\frac{\FTD{\ftd}}{\ftd^{1-\alpha}} \Rightarrow 
	\begin{cases}
        	K(r) & {\rm if} \ \alpha = 0, \\ 
        	A^{\phi_0} (r) & {\rm if} \ 0< \alpha < 1, 
	\end{cases} 
\end{equation}
and if $\alpha=1$, then 
\begin{equation}
	\Pr ( \FTD{\ftd} = x ) \sim \frac{1+(-1)^{n+x}}{2} \mathcal{J}(x;r). 
\end{equation}
and if $\alpha>1$, then 
\begin{equation}
	\Pr (\FTD{\ftd} =x) \to \df (x)
\end{equation}
as $\ftd \to \infty$. 
\end{col}
\begin{proof}
Since the $\alpha=0$ case corresponds to the DTQW, $\FTD{\ftd}/\ftd \Rightarrow K(r)$ as $\ftd \to \infty$. 
The $\sqrt{r(\ftd)}=\ftd^\alpha$ with $0 < \alpha < 1$ satisfies the condition $\ftd \sqrt{r(n)} \to 0$ as 
$\ftd \to \infty$. Therefore, it follows from Theorem~\ref{COQCD} that 
$\FTD{\ftd}/\ftd^{1-\alpha} \Rightarrow A^{\phi_0}$ with $0<\alpha<1$. 
The result on $\alpha \geq 1$ case derives from Theorem~\ref{kamakura}. 
\end{proof}
\section{One-dimensional Discrete Time Quantum Walk as Quantum Simulator} \label{qs_sec}
\begin{figure}[ht]
	\centering
	\includegraphics[width=13cm]{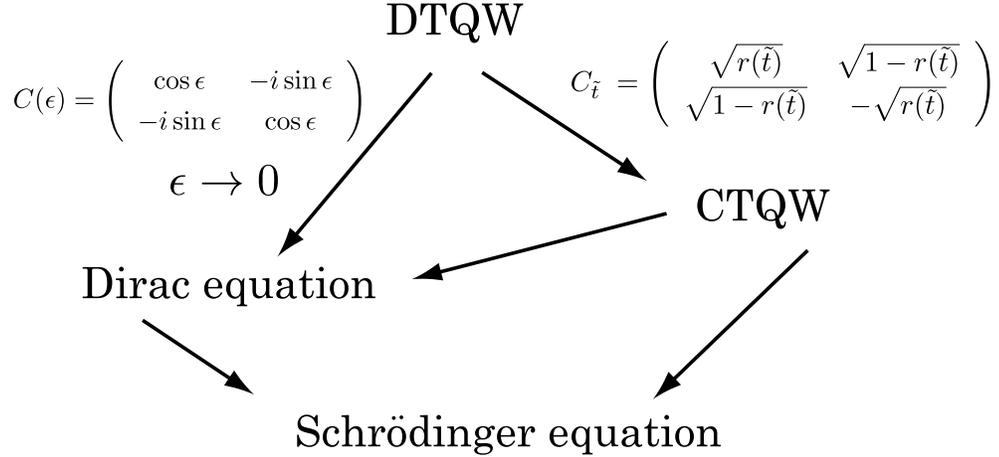}
	\caption{The relationship among the DTQW, the CTQW, and the Dirac and Schr\"{o}dinger equations on 
	a one-dimensional space for the asymptotic behaviors. The DTQW can derive the Dirac equation and 
	the CTQW. These can derive the Schr\"{o}dinger equation. It is noted that the CTQW on the graph 
	can derive the 1+1 dimensional Dirac equation~\cite{Childs}.}
	\label{fromdtqw}
\end{figure}
Summarizing the above results, the DTQW can simulate the Dirac equation and the CTQW as the asymptotic 
behavior on the one-dimensional space. While the DTQW is run under non-relativistic quantum mechanics, 
the DTQW can simulate the relativistic situation. Therefore, the DTQW can be taken 
as a quantum simulator, which is different from the original idea of Richard 
Feynman~\cite{FEYNMAN}. Realizing the DTQW means that such experimental setup can 
test the quantum world. Compared to the original theory, the parameters 
of the quantum coin and the shift in the real space may be connected to the parameters and 
the constant number in the original theory. This means that we can construct the alternative 
theory beyond the quantum theory inside quantum theory. It should be emphasized that this theory 
is not realized but is consistent and can be interpreted. However, the DTQW has the potential 
to the test of the physical constant by changing this as the parameter in the alternative theory. 
On the other words in mathematics, the DTQW can derive the Dirac equation, the CTQW, 
and the Schr\"{o}dinger equation. Here, the Schr\"{o}dinger equation can be 
derived from the approximation of the Dirac equation and the CTQW. Therefore, the DTQW is 
mathematically a superordinate concept. These relationships can be illustrated in Fig.~\ref{fromdtqw}.
\section*{Acknowledgment}
The author acknowledges useful collaborations and discussion with Etsuo Segawa, Kota Chisaki and Norio Konno. The author also thanks Hosho Katsura and Takuya Kitagawa for useful discussion. The author would like to thank the use of the utilities of Tokyo Institute of Technology and Massachusetts Institute of Technology and many technical and secretary supports. The author is grateful to the financial supports from JSPS Research Fellowships for Young Scientists (No. 21008624), JSPS Excellent Young Researcher Overseas Visit Program, Global Center of Excellence Program ``Nanoscience and Quantum Physics" at Tokyo Institute of Technology during his Ph.D study.
\appendix
\section{Classical stochastic process}
In this appendix, we define the stochastic process as 
\begin{defin}[Stochastic process]
	The stochastic process is defined as 
	\begin{equation}
		\{ X_t (\omega), t \in T \} \ \omega \in \Omega,
	\end{equation}
	where $T$ means time range and is assumed as the natural number $T = \N$ to 
	simplify the discussion.
\end{defin}
The simplest example is the Markov process, which is defined that $X_{t+1}$ is independent of 
$X_t$ for any $t \in T$. One of the Markov process is the random walk (RW). 
\begin{defin}[Time-independent and spatial-independent one-dimensional random walk]
	The time-independent and spatial-independent one-dimensional random walk is defined as the 
	independent and identical distributed (i.i.d.) sequence $X_t$ to satisfy  
	\begin{align}
		\Pr [X_{t+1} = x ] & = p \Pr [X_t = (x -1) ] + (1 - p) \Pr [X_t = (x + 1) ], \notag \\
		\Pr [X_0 = x] & = \delta_{x,0}
	\end{align}
	for any $t \in \N$ and $x \in \Z$.
\end{defin}

To describe the asymptotic behavior for the random variable $X_n$ with the suffix $n \in \N$, 
we define the weak convergence of distribution~\footnote{This is also called the 
weak convergence, convergence of distribution, or the probability convergence.} as 
\begin{defin}[Weak convergence of distribution]
	Let $\{ X_n \}$ be the sequence of the random variables. The random variable $X_n$ weakly 
	converges to the random variable $Y$ denoted as 
	\begin{equation}
		X_n \Rightarrow Y
	\end{equation}
	if and only if, for any $\epsilon >0$, 
	\begin{equation}
		\Pr [|X_n - Y|> \epsilon] \to 0 \ \ {\rm as} \ n \to \infty.
	\end{equation}
\end{defin}
This definition can be expressed from the viewpoint of the probability measure as 
\begin{equation}
	\lim_{n \to \infty} \int_{\omega \in \Omega} f(X_n [\omega]) dP_n [\omega] = 
	\int_{\omega \in \Omega} f(Y[\omega]) dP [\omega], 
\end{equation}
where $dP_n$ and $dP$ are defined in the random variable $X_n$ and $Y$, respectively.
In the case of the time-independent and spatial-independent one-dimensional RW, we obtain 
the asymptotic behavior in the following: 
\begin{thm}[Central limit theorem~\footnote{This is often called the de Moivre-Laplace theorem. 
The general version of this theorem to satisfy Eq.~(\ref{centraleq}) 
for the independent sequence $\{ Y_n \}$ was 
proven by Lindeberg to see more details in the book~\cite{rw_book}.}] \label{central}
	Let $\{ Y_n \}_{n \in \N}$ be the i.i.d. sequence to satisfy 
	\begin{equation}
		\Pr [Y_n = 1] = p \ \ \ \ \Pr [Y_n = 1] = 1 - p
	\end{equation}
	for any $n \in \N$. Put $S_m = \sum_{k = 1}^{m} Y_k$. Then, we obtain 
	\begin{equation} \label{centraleq}
		\frac{S_m - \Ex [S_m]}{\sqrt{\Var[S_m]}} \Rightarrow N(0,1) \ \ {\rm as} \ m \to \infty,
	\end{equation}
	where $N(\mu,\nu)$ is the random variable to express the normal distribution~\footnote{In physics, 
	this distribution is called the Gaussian.} with the mean $\mu$ and the variance $\nu$ given by 
	\begin{equation} \label{normaleq}
		\Pr [N(\mu,\nu) = x] = \frac{1}{\sqrt{2 \pi \nu}} \Exp \left[ - \frac{(x - \mu)^2}{2 \nu} \right].
	\end{equation}
	Here, the expectation value and the variance are given by 
	\begin{align}
		\Ex [S_m] & = m p, \\
		\Var [S_m] & = m p (1-p).
	\end{align}
\end{thm}
It is noted that the normal distribution (\ref{normaleq}) is characterized only by the mean and 
the variance. 
\begin{proof}
	Given the function $T_m$ as 
	\begin{equation}
		T_m := \sum_{k = 1}^{m} \frac{Y_k - p}{\sqrt{m p (1 - p)}},
	\end{equation}
	we obtain the moment-generating function for $T_m$ as 
	\begin{align}
		\Ex \left[e^{i \tau T_m} \right] & = \prod_{k = 1}^{m} \Ex \left[ \Exp \left( \frac{i \tau 
		(Y_k - p)}{\sqrt{m p (1 - p)}} \right) \right] \notag \\
		 & = \left[ p \Exp \left( i \tau \sqrt{\frac{(1 - p)}{m p}} \right) 
			+ (1 - p) \Exp \left( - i \tau \sqrt{\frac{p}{m (1 - p)}} \right) \right]^m \notag \\
		 & = \left[ p \left( 1 + i \tau \sqrt{\frac{(1 - p)}{m p}} - \frac{(1 - p)}{2 m p} \tau^2 \right) 
		 \right. \notag \\ & \ \ \ \ \ \ \ \ \left.
		 + (1 - p) \left( 1 - i \tau \sqrt{\frac{p}{m (1 - p)}} - \frac{p}{2 m (1 - p)} \tau^2 \right) 
		 + O \left( \frac{1}{m^2} \right) \right]^m \notag \\
		 & = \left[ 1 - \frac{\tau^2}{2 m} + O \left( \frac{1}{m^2} \right) \right]^m \notag \\
		 & \to \Exp \left( - \frac{\tau^2}{2} \right) \ \ {\rm as} \ m \to \infty.
	\end{align}
	From the inverse Fourier transform for the last line, we obtain the desired result. 
\end{proof}
Applying the central limit theorem to the time-independent and spatial-independent one-dimensional RW, 
we obtain the following corollary by replacing $S_m$ to $X_t$:
\begin{col}
	The time-independent and spatial-independent one-dimensional RW $X_t$ has the asymptotic behavior as 
	\begin{equation}
		\frac{X_t}{\sqrt{t}} \Rightarrow N \left( 2p-1, 4p (1-p) \right).
	\end{equation}
	Furthermore, in case of the unbiased coin, $p = 1/2$, one has 
	\begin{equation}
		\frac{X_t}{\sqrt{t}} \Rightarrow N \left( 0, 1 \right).
	\end{equation}
\end{col}
\begin{proof}
	We calculate the expectation value of $X_t$ as 
	\begin{equation}
		\Exp [X_t] = \sum_{k = 1}^{t} (p - (1 -p)) = t (2 p -1).
	\end{equation}
	The variance is calculated as 
	\begin{equation}
		\Var [X_t] = \sum_{k = 1}^{t} [p + (1 -p)] - (2 p -1)^2  = 4 t p (1 - p).
	\end{equation}
	Applying them to Theorem~\ref{central}, we obtain the desired result.
\end{proof}


\begin{thebibliography}{999}
\bibitem{Abal07}
	G. Abal, R. Donangelo, and H. Fort,
	in {\it Annals of the 1st Workshop on Quantum Computation and Information}, p. 189, 
	UCPel, 9-11 October. 2006, Pelotas, RS, Brazil, arXiv:0709.3279.
\bibitem{AharonovQW}
	Y. Aharonov, L. Davidovich, and N. Zagury,
	Phys. Rev. A
	{\bf 48}, 1687 (1993).
\bibitem{Ahlbrecht11}
	A. Ahlbrecht, A. Alberti, D. Meschede, V. B. Scholz, A. H. Werner, and R. F. Werner, 
	arXiv:1105.1051.
\bibitem{Ambainis}
	A. Ambainis, E. Bach, A. Nayak, A. Vishwanath, and J. Watrous, 
	in {\it Proceedings of the 33rd Annual ACM Symposium on Theory of Computing (STOC'01)}
	(ACM Press, New York, 2001), p. 37.
\bibitem{Ambainis2}
	A. Ambainis, 
	SIAM Journal on Computing
	{\bf 37}, 210 (2007).
\bibitem{Bracken}
	A. J. Bracken, D. Ellinas, and I. Smyrnakis,
	Phys. Rev. A
	{\bf 75}, 022322 (2007).
\bibitem{Broome}
	M. A. Broome, A. Fedrizzi, B. P. Lanyon, I. Kassal, A. Aspuru-Guzik, and A. G. White,
	Phys. Rev. Lett. 
	{\bf 104}, 153602 (2010).
\bibitem{Buhrman}
	H. Buhrman and R. Spalek, 
	in {\it Proc. 17th ACM-SIAM Symposium on Discrete Algorithms}
	(Society for Industrial and Applied Mathematics, Philadelphia, 2006), p. 880.
\bibitem{Chandrashekar06}
	C. M. Chandrashekar, 
	arXiv:quant-ph/0609113.
\bibitem{Chandrashekar2}
	C. M. Chandrashekar,
	Cent. Eur. J. Phys.,
	{\bf 8}, 979 (2010).
\bibitem{Chandrashekar10}
	C. M. Chandrashekar, S. Banerjee, and R. Srikanth, 
	Phys. Rev. A
	{\bf 81}, 062340 (2010).
\bibitem{Chandrashekar08}
	C. M. Chandrashekar and R. Laflamme, 
	Phys. Rev. A
	{\bf 78}, 022314 (2008).
\bibitem{Childs2}
	A. M. Childs,
	Phys. Rev. Lett.
	{\bf 102}, 180501 (2009).
\bibitem{Childs}
	A. M. Childs and J. Goldstone, 
	Phys. Rev. A 
	{\bf 70}, 042312 (2004).
\bibitem{Chisaki}
	K. Chisaki, M. Hamada, N. Konno, and E. Segawa,
	Interdisciplinary Information Sciences
	{\bf 15}, 423 (2009).
\bibitem{Chisaki2}
	K. Chisaki, N. Konno, and E. Segawa,
	Quant. Inf. Comp. 
	{\bf 12}, 0314 (2012).
\bibitem{CKSS2}
	K. Chisaki, N. Konno, E. Segawa, and Y. Shikano, 
	Quant. Inf. Comp.
	{\bf 11}, 0741 (2011).
\bibitem{Eckert}
	K. Eckert, J. Mompart, G. Birkl, and M. Lewenstein, 
	Phys. Rev. A 
	{\bf 72}, 012327 (2005).
\bibitem{Farhi}
	E. Farhi and S. Gutmann,
	Phys. Rev. A
	{\bf 58}, 915 (1998).
\bibitem{FEYNMAN}
	R. P. Feynman,
	Int. J. Theor. Phys. 
	{\bf 21}, 467 (1982).
\bibitem{random_book}
	P.-G. D. Gennes,
	{\it Scaling Concepts in Polymer Physics}
	(Cornell University Press, New York, 1979).
\bibitem{Roos2}
	R. Gerritsma, G. Kirchmair, F. Z\"{a}hringer, E. Solano, R. Blatt, and C. F. Roos,
	Nature
	{\bf 463}, 68 (2010).
\bibitem{GASM}
	M. G\"{o}n\"{u}lol, E. Ayd{\i}ner, Y. Shikano, and \"{O}. E. M\"{u}stecapl{\i}o\~{g}lu, 
	New J. Phys. 
	{\bf 13}, 033037 (2011).
\bibitem{Grimmett}
	G. Grimmett, S. Janson, and P. F. Scudo, 
    Phys. Rev. E.
    {\bf 69}, 026119 (2004).
\bibitem{Gudder}
	S. P. Gudder, 
	{\it Quantum Probability}
	(Academic Press Inc., San Diego, 1988).
\bibitem{Hasebe}
	T. Hasebe,
	arXiv:1009.1505.
\bibitem{Inui3}
	N. Inui, Y. Konishi, and N. Konno,
	Phys. Rev. A
	{\bf 69}, 052323 (2004).
\bibitem{Inui2}
	N. Inui, and N. Konno,
	Physica A
	{\bf 353}, 133 (2005).
\bibitem{Inui}
	N. Inui, N. Konno, and E. Segawa, 
	Phys. Rev. E
	{\bf 72}, 056112 (2005).
\bibitem{Joye2}
	A. Joye,
	Commun. Math. Phys.
	{\bf 307}, 65 (2011).
\bibitem{Joye}
	A. Joye and M. Merkli,
	J. Stat. Phys.
	{\bf 140}, 1 (2010).
\bibitem{Karski}
	M. Karski, L. F\"{o}ster, J.-M. Choi, A. Steffen, W. Alt, D. Meschede, and A. Widera, 
	Science
	{\bf 325}, 174 (2009).
\bibitem{Kempe03}
	J. Kempe,
	Contemp. Phys.
	{\bf 44}, 307 (2003).
\bibitem{Kendon07}
	V. Kendon,
	Math. Struct. in Comp. Sci. 
	{\bf 17}, 1169 (2007).
\bibitem{Kitagawa}
	T. Kitagawa, 
	to be published from Quant. Inf. Proc., 
	arXiv:1112.1882.
\bibitem{Kitagawa11}
	T. Kitagawa, M. A. Broome, A. Fedrizzi, M. S. Rudner, E. Berg, I. Kassal, 
	A. Aspuru-Guzik, E. Demler, and A. G. White,
	Nat. Comm. 
	{\bf 3}, 882 (2012).
\bibitem{Kitagawa10}
	T. Kitagawa, M. S. Rudner, E. Berg, and E. Demler,
	Phys. Rev. A 
	{\bf 82}, 033429 (2010).
\bibitem{Konno02} 
	N. Konno, 
	Quantum Inf. Proc.
	{\bf 1}, 345 (2002).
\bibitem{Konno05} 
	N. Konno, 
	J. Math. Soc. Jpn.
	{\bf 57}, 1179 (2005).
\bibitem{KonnoQLC} 
	N. Konno, 
	Phys. Rev. E 
	{\bf 72}, 026113 (2005).
\bibitem{KonnoRev}
	N. Konno, 
	in {\it Quantum Potential Theory}, 
	Lecture Notes in Mathematics Vol. 1954, 
	edited by U. Franz and M. Schurmann 
	(Springer-Verlag, Heidelberg, 2008), p.309.
\bibitem{KonnoCor}
	N. Konno, 
	Stochastic Models 
	{\bf 25}, 28 (2009). 
\bibitem{Konno1}
	N. Konno,
	Quantum Inf. Proc.
	{\bf 8}, 387 (2009).
\bibitem{Konno2}
	N. Konno,
	Quantum Inf. Proc.
	{\bf 9}, 405 (2010).
\bibitem{Konno}
	N. Konno and E. Segawa,
	Quant. Inf. Comp.
	{\bf 11}, 0485 (2011).
\bibitem{Lahini11}
	Y. Lahini, M. Verbin, S. D. Huber, Y. Bromberg, R. Pugatch, and Y. Silberberg, 
	Phys. Rev. A 
	{\bf 86}, 011603 (2012).
\bibitem{Linden}
	N. Linden and J. Sharam,
	Phys. Rev. A
	{\bf 80}, 052327 (2009).
\bibitem{Liu}
	C. Liu and N. Petulante,
	Phys. Rev. A 
	{\bf 79}, 032312 (2009).
\bibitem{Liu11}
	C. Liu and N. Petulante, 
	Phys. Rev. A
	{\bf 84}, 012317 (2011).
\bibitem{Lovett}
	N. B. Lovett, S. Cooper, M. Everitt, M. Trevers, and V. Kendon, 
	Phys. Rev. A 
	{\bf 81}, 042330 (2010).
\bibitem{Mackay}
	T. D. Mackay, S. D. Bartlett, L. T. Stephenson, and B. C. Sanders,
	J. Phys. A
	{\bf 35}, 2745 (2002).
\bibitem{Magniez}
	F. Magniez and A. Nayak, 
	Algorithmica
	{\bf 48}, 221 (2007).
\bibitem{Magniez2}
	F. Magniez, M. Santha, and M. Szegedy, 
	SIAM Journal of Computing
	{\bf 37}, 413 (2007).
\bibitem{Maloyer07}
	O. Maloyer and V. Kendon, 
	New J. Phys. 
	{\bf 9}, 87 (2007).
\bibitem{Mayer11}
	K. Mayer, M. C. Tichy, F. Mintert, T. Konrad, and A. Buchleitner, 
	Phys. Rev. A 
	{\bf 83}, 062307 (2011).
\bibitem{McGettrick}
	M. McGettrick, 
	Quant. Inf. Comp. 
	{\bf 10}, 0509 (2010).
\bibitem{Meyer}
	D. Meyer,
	J. Stat. Phys.
	{\bf 85}, 551 (1996).
\bibitem{Obuse}
	H. Obuse and N. Kawakami,
	Phys. Rev. B
	{\bf 84}, 195139 (2011).
\bibitem{Oka05}
	T. Oka, N. Konno, R. Arita, and H. Aoki,
	Phys. Rev. Lett.
	{\bf 94}, 100602 (2005).
\bibitem{Peruzzo}
	A. Peruzzo, M. Lobino, J. C. F. Matthews, N. Matsuda, A. Politi, K. Poulios, X. Zhou, 
	Y. Lahini, N. Ismail, K. W\"{o}rhoff, Y. Bromberg, Y. Silberberg, M. G. Thompson and J. L. O'Brien, 
	Science 
	{\bf 329}, 1500 (2010).
\bibitem{rw_book}
	P. Rev\'{e}sz, 
	{\it Random walk in random and non-random environments}
	(World Scientific, Singapore, 1990).
\bibitem{Rohde11}
	P. P. Rohde, A. Fedrizzi, and T. C. Ralph, 
	J. Mod. Opt.
	{\bf 59}, 710 (2012).
\bibitem{Romanelli}
	A. Romanelli, 
	Physica A 
	\textbf{388}, 3985 (2009).
\bibitem{Romanelli1}
	A. Romanelli, 
	Phys. Rev. A 
	{\bf 80}, 042332 (2009).
\bibitem{Ryan}
	C. A. Ryan, M. Laforest, J. C. Boileau, and R. Laflamme, 
	Phys. Rev. A
	{\bf 72}, 062317 (2005).
\bibitem{Sato}
	F. Sato and M. Katori, 
	Phys. Rev. A 
	{\bf 81}, 012314 (2010).
\bibitem{Schreiber}
	A. Schreiber, K. N. Cassemiro, V. Poto\v{c}ek, A. G\'{a}bris, P. J. Mosley, E. Andersson, I. Jex, and 
	Ch. Silberhorn, 
	Phys. Rev. Lett. 
	{\bf 104}, 050502 (2010).
\bibitem{Shenvi}
	N. Shenvi, J. Kempe, and K. Whaley, 
	Phys. Rev. A
	{\bf 67}, 052307, (2003).
\bibitem{CKSS1}
	Y. Shikano, K. Chisaki, E. Segawa, and N. Konno,
	Phys. Rev. A
	{\bf 81}, 062129 (2010). 
\bibitem{SK}
	Y. Shikano and H. Katsura,
	Phys. Rev. E
	{\bf 82}, 031122 (2010).
\bibitem{SK2}
	Y. Shikano and H. Katsura,
	AIP Conf. Proc.
	{\bf 1363}, 151 (2011).
\bibitem{Stefanak11}
	M. \v{S}tefa\v{n}\'{a}k, S. M. Barnett, B. Koll\'{a}r, T. Kiss, and I. Jex, 
	New J. Phys. 
	{\bf 13}, 033029 (2011).
\bibitem{Stefanak}
	M. \v{S}tefa\v{n}\'ak, I. Jex, and T. Kiss,
	Phys. Rev. Lett. 
	{\bf 100}, 020501 (2008).
\bibitem{Stefanak2}
	M. \v{S}tefa\v{n}\'ak, T. Kiss, and I. Jex,
	Phys. Rev. A
	{\bf 78}, 032306 (2008).
\bibitem{Stefanak3}
	M. \v{S}tefa\v{n}\'ak, T. Kiss, and I. Jex,
	New J. Phys.
	{\bf 11}, 043027 (2009).
\bibitem{Strauch1}
	F. W. Strauch, 
	Phys. Rev. A
	{\bf 74}, 030301 (2006). 
\bibitem{Strauch}
	F. W. Strauch,
	J. Math. Phys.
	{\bf 48}, 082102 (2007).
\bibitem{Tregenna}
	B. Tregenna, W. Flanagan, R. Maile, and V. Kendon,
	New J. Phys.
	{\bf 5}, 83 (2003).
\bibitem{VA}
	S. E. Venegas-Andraca, 
	{\it Quantum Walks for Computer Scientists} 
	(Morgan and Claypool, 2008).
\bibitem{VA2}
	S. E. Venegas-Andraca, 
	to be published from Quan. Inf. Proc., arXiv:1201.4780.
\bibitem{Watabe}
	K. Watabe, N. Kobayashi, M. Katori, and N. Konno,
	Phys. Rev. A
	{\bf 77}, 062331 (2008).
\bibitem{Wojcik}
	A. W\'ojcik, T. {\L}uczak, P. Kurzy\'nski, A. Grudka, and M. Bednarska,
	Phys. Rev. Lett. 
	{\bf 93}, 180601 (2004).
\bibitem{Xu}
	X.-P. Xu,
	Eur. Phys. J. B 
	{\bf 77}, 479 (2010).
\bibitem{Roos}
	F. Z\"{a}hringer, G. Kirchmair, R. Gerritsma, E. Solano, R. Blatt, and C. F. Roos,
	Phys. Rev. Lett. 
	{\bf 104}, 100503 (2010).
\end{thebibliography}
\end{document}